\documentclass{article}
\usepackage{amsmath}
\usepackage{amsthm}
\usepackage{amssymb}
\usepackage{bbm}
\usepackage{amsfonts}
\usepackage{graphicx}
\usepackage{enumerate}
\usepackage{color}
\usepackage{cleveref}
\usepackage{url}
\usepackage[normalem]{ulem}

\newtheorem{theorem}{Theorem}[section]
\newtheorem{definition}[theorem]{Definition}

\newtheorem{corollary}[theorem]{Corollary}
\newtheorem{lemma}[theorem]{Lemma}

\newtheorem{remark}[theorem]{Remark}
\newcommand{\ed}{\mathrm{d}}

\definecolor{darkgreen}{RGB}{0,127,0}

\renewcommand{\P}{{\mathbb P}}
\newcommand{\Q}{{\mathbb Q}}
\newcommand{\E}{{\mathbb E}}
\newcommand{\M}{{{\cal M}_1}({\mathbb R})}
\newcommand{\R}{{\mathbb R}}
\newcommand{\maximizeterm}{\sup}

\title{Optimizing S-shaped utility and implications for risk management\thanks{We are grateful to Doctor Dirk Tasche for helpful suggestions and discussion on the first version. We thank Professor Xunyu Zhou who, during our seminar at Columbia University in NY on November 29 2017, provided us with a number of references we had missed in our literarure review and that have been included in the introduction of this updated version.}}
\author{John Armstrong\\ Dept. of Mathematics \\ King's College London \\ {\tt \small john.1.armstrong@kcl.ac.uk} \and Damiano Brigo \\ Dept. of Mathematics \\ Imperial College London \\ {\tt \small damiano.brigo@imperial.ac.uk}}

\date{{\small First published on arXiv.org on Nov 3, 2017, arXiv 1711.00443. This version: \today}}

\begin{document}

\maketitle

\vspace{-0.5cm}

\begin{abstract}
We consider market players with tail-risk-seeking behaviour as exemplified by the S-shaped utility introduced by Kahneman and Tversky. We argue that risk measures such as value at risk (VaR) and expected shortfall (ES)  are ineffective in constraining such players.  We show that, in many standard market models, product design aimed at utility maximization is not constrained at all by VaR or ES bounds: the maximized utility corresponding to the optimal payoff is the same with or without ES constraints. By contrast we show that, in reasonable markets, risk management constraints based on a second more conventional concave utility function can reduce the maximum S-shaped utility that can be achieved by the investor, even if the constraining utility function is only rather modestly concave. It follows that product designs leading to unbounded S-shaped utilities will lead to unbounded negative expected constraining utilities when measured with such conventional utility functions. To prove these latter results we solve a general problem of optimizing an investor expected  utility under risk management constraints where both investor and risk manager have conventional concave utility functions, but the investor has limited liability. We illustrate our results throughout with the example of the Black--Scholes option market.
These results are particularly important given the historical role of VaR and that ES was endorsed by the Basel committee in 2012--2013.
\end{abstract}

%

{\small {\bf Keywords and phrases}: Optimal product design under risk constraints; value at risk constraints; expected shortfall constraints; concave utility constraints; S-shaped utility maximization;  limited liability investors; tail risk seeking investors; effective risk constraints; concave utility risk constraints.}

\medskip

{\small \noindent {\bf AMS} classification codes: 91B16, 91B28, 91B30, 91B70; {\bf JEL}: D81, G11, G13.

\newpage

\section{Introduction}

In this paper we consider market players with tail-risk-seeking behaviour as exemplified by the S-shaped utility introduced by Kahneman and Tversky \cite{kahnemanAndTversky}. We argue that risk measures such as value at risk (VaR) and
 expected shortfall (ES) (also known as conditional value at risk (CVaR) or more rarely as average value at risk (AVaR)) are ineffective in constraining such players. To illustrate this we show that, in many familiar market models, product design aimed at utility maximization for tail-risk-seeking market players is unaffected by ES bounds. Given that for a fixed confidence level ES dominates VaR, the analysis under VaR constraints is completely analogous and easier, so in the paper we focus on ES.  To show this ineffective behavior, we prove that the maximized utility corresponding to the optimal payoff in the optimization problem is the same with or without an ES constraint. This is particularly important in the light of the fact that ES has been officially endorsed and suggested as a risk measure by the Basel committee in 2012-2013 \cite{basel2013,Basel_Minimum_Capital_MR}, partly for its ``coherent risk measure'' properties \cite{artznerEtAl,acerbitasche}. 

We are not the first to criticize VaR and ES. A full literature review on criticism of VaR and ES is beyond the scope of this introduction. We recall the above-mentioned works \cite{artznerEtAl,acerbitasche}, where VaR is criticized for its lack of coherence and sub-additivity more in particular. However, ES has not been immune from criticism either.  \cite{danielsson}  for example argues that the reliability of any backtesting procedure for ES is much lower than that of VaR. The issue of backtestability of ES has been discussed via the elicitability debate started in \cite{gneiting} in recent years, although \cite{acerbi} clarified several issues on the matter, arguing that ES is backtestable, see also \cite{tasche}  and \cite{ziegel} where it is proven that the pair formed by VaR and ES is jointly elicitable. The paper \cite{davis}, working under the framework of prequential statistics, argues that VaR has important advantages over ES in terms of verification properties. The debate is continuing to this day.

However, our criticism here does not center on a choice between VaR and ES and is of a different, more fundamental nature. Our criticism links directly with the use of risk measures as an excessive risk-taking control tool, where we show that VaR and ES clearly fail in a quite general market model. We further illustrate this failure by specializing our result to a Black--Scholes option market setting, given that the Black--Scholes market is a typical benchmark case.  

The first part of the paper is thus a negative result on the use of expected shortfall for curbing excessive tail-risk-seeking behaviour. The natural question is which alternatives could work? 

In the second part of the paper we introduce a possible solution. We calculate the solution of the payoff-design optimization problem with expected utility constraints replacing ES constraints. In this case two utilities will be involved: the investor utility to be maximized, and the risk manager's concave utility that will constrain the strategy. We are able to calculate the optimal strategy in a special case corresponding to a conventionally risk-averse investor with limited liability.  We will see that in this case risk constraints are effective in curbing excessive tail-risk-seeking behaviour under quite modest conditions on the market and the risk manager's utility function.

Note that it follows from the last result that in order to achieve unboundedly positive expected utilities measured with the investor's S-shaped utility function, the expected utilities measured using the risk manager's concave utility function must be unboundedly negative.

We now present a literature review of earlier related work\footnote{Following a presentation of the first version of this paper  (published online on November 3, 2017) at Columbia University in New York on November 29, 2017, we learned that extensive work had been done previously on related themes that partly overlaps with this paper.}. 

Utility optimization under risk measure constraints was considered earlier in \cite{basak}, who adopt a framework that is exactly the one we adopt in this paper, 
but only under standard utility assumptions and no S-shaped utility in particular. In that paper it is shown that in the case where a large loss occurs, it is an even larger loss under value at risk based risk management. This occurs because when the constraint is binding, a market player who is forced by the VaR constraint to reduce portfolio losses in some states would finance these reduced losses by increasing portfolio losses in the costly states where the terminal state price density is large. As such states already have the lowest terminal portfolio value for the unconstrained problem, the VaR constraint ends up fattening the left tail of the terminal portfolio distribution. This leads to increased probability of extreme losses. 

In \cite{cuoco} it is shown that VaR constraints play a better role when, as is done in practice,  the portfolio VaR is re-evaluated dynamically by incorporating available conditioning information.  Again, this is done under standard utility and S-shaped utility is not considered. 

Prospect theory has been studied in relation to risk measures and portfolio choice in a series of papers by Xunyu Zhou and co-authors. We will consider primarily the four papers \cite{xyz1,xyz3,xyz2,xyz4}. Related research is presented in 
\cite{xyz5,xyz6,xyz7,xyz8,xyz9,xyz10}. 

The first paper we consider here is the 2008 paper \cite{xyz1}, where optimal portfolio choice under S-shaped utility and probability distorsions is considered and solved. This is done without a risk constraint, though, but essentially under a budget constraint and at times a no-bankruptcy constraint, meaning a positive terminal value. Techniques include essentially what we call $X$-rearrangements here in our paper and connections with the classic theory of re-arrangements and inequalities by Hardy and Littlewood is done explicitly in the  2010 paper \cite{xyz3}.

The 2011 paper \cite{xyz2} generalizes and abstracts the approach in \cite{xyz1}, addressing a variety of models, leading to the ``quantile formulation''. It solves a general problem of utility maximization, including S-shaped utility with probability distorsion and under a budget constraint but no risk constraint. This 2011 paper uses law invariance as we do here and adopts similar assumptions and techniques. In particular, it is found that the optimal terminal payment is anti-comonotonic with the pricing kernel. If one ignores the risk management constraint, the 2011 paper effectively contains a proof of our Theorem  \ref{theorem:rearrangement} below in a more general setting.\footnote{We had developed our proof independently and found out later that a similar proof for the problem without risk constraint had been published earlier in \cite{xyz2} in a more general setting.} 
However, we find that although the techniques and the discussion in that paper anticipate many of the techniques and ideas we use here, again portfolio choice based on S-shaped utility maximization under a value at risk or expected shortfall constraint is not addressed, as further confirmed by the five motivating models in \cite{xyz2}. 

The 2015 paper \cite{xyz4} considers the problem of miminizing a risk measure (typically generalizations of weighted Value at Risk, WVaR) under a minimum performance constraint, or on optimizing mean risk (as opposed to expected utility). There are clear connections with our paper. The closest we get, it seems to us, is in Section 7.2 of \cite{xyz4} where VaR constraints are used but expected utility is not considered as the objective to be maximized, as the expected return or mean risk is maximized instead. We note that Section 7.2 of \cite{xyz4} is presented as a contribution to the debate on portfolio choice under risk constraints as for example in the above-mentioned earlier paper \cite{cuoco}, although while \cite{cuoco} does consider expected (but not S-shaped) utility, \cite{xyz4} does not do so in Section 7.2. 

We remain confident that optimal portfolio choice under s-shaped utility and VaR/ES constraints or concave utility risk constraints  remains therefore an original problem, although the above papers represent a fundamental contribution to behavioural finance,  prospect theory, portfolio choice and risk measures, anticipating part of our results and techniques.

The paper is organized as follows. In Section \ref{sec:Sshaped} we introduce S-shaped utility and define the notion of a tail-risk-seeking market player. We also point out that utility has to be measured on the right variables, like the pay-packet of a trader for example rather than simply his portfolio P\&L, and that in such cases S-shaped utility does not necessarily denote irrational behaviour. 

In Section \ref{sec:lawinvariant} we introduce the optimization problem that will be used in the paper. We find a given market optimal payoff design (simple claim) that 
maximizes a function of the claim distribution under a price constraint and a risk measurement constraint. The risk management constraint is based on the claim probability distribution and could be for example a limit on VaR or ES of the claim position.  
The solution of this problem is given in a theorem, proved in Appendix \ref{appendix:rearrangementProof} where we use a notion of rearrangement similar to that used in the Hardy and Littlewood inequality for symmetric decreasing rearrangements, see also the earlier works \cite{dybvig} and \cite{xyz3}\footnote{While we have developed the connection with the theory of Hardy and Littlewood independently, we found out later that this had been realized earlier in \cite{xyz3}, see also \cite{dybvig}.}. Incidentally, the fact that we have both prices and risk measures leads us to model explicitly also the Radon--Nykodim derivative linking the pricing measure and the physical measure. 

In Section \ref{sec:optimiz} we apply our optimization result to payoff design optimization under S-shaped utility and expected shortfall risk constraints. This is the main negative result of the paper, where we show that the expected shortfall constraint is irrelevant in that the problem has the same optimal expected utility with or without it. 

In Section \ref{section:blackScholesDiscrete} we apply our result in the Black Scholes market, showing the specific calculations and entities that arise in that setting. 

Section \ref{sec:limlam} starts looking at effective solutions for curbing excessive risk-seeking behaviour that do not suffer from the ES problems. We see that an effective constraint can be based on putting bounds on expected concave utilities. In this case the risk constraint, based on the expected concave utility for losses, impacts the optimal payoff design. The objective function is still based on the expected utility corresponding to gains. We still work under an additional price/budget constraint. We also provide a two step algorithm that allows one to compute the optimal solution via line search. We conclude by illustrating calculations for a specific parametric choice of the utility function and market and show that for strategies that attain infinite S-shaped investor expected utility, the risk expected concave utility constraint is infinitely bad.

\section{S-shaped utility and tail-risk-seeking behaviour}\label{sec:Sshaped}

In \cite{kahnemanAndTversky}, Kahneman and Tversky observed that individuals
appear to have preferences governed by an S-shaped utility function. By ``S-shaped'' utility Kahneman and Tversky mean a number of things.
\begin{enumerate}[(i)]
\item  They are increasing
\item  They are strictly convex on the left
\item  They are strictly concave on the right
\item  They are non-differentiable at the origin
\item  They are asymmetrical: negative events are considered worse than positive events are considered good.
\end{enumerate}
A typical S-shaped utility function is shown in Figure \ref{fig:sshapedcurve}.
\begin{figure}[htpb]
\centering
\includegraphics[width=0.7\linewidth]{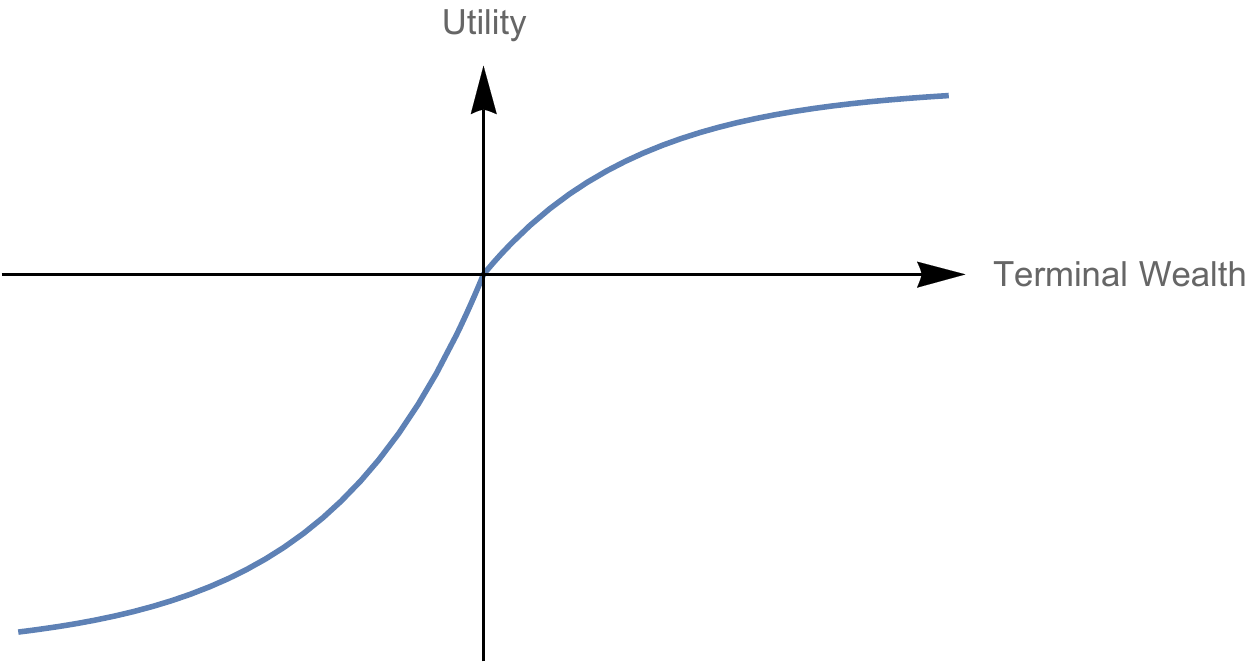}
\caption{An S-Shaped utility function.}
\label{fig:sshapedcurve}
\end{figure}
The prototypical example of S-shaped utility (see for example \cite{follmerBook}, Formula 2.9) is
\begin{equation}\label{eq:ktexample} u(x) = x^\gamma 1_{\{ x \ge 0\}}  - \lambda (-x)^\gamma 1_{\{ x <  0\}} , 
\end{equation}
for a zero benchmark level, with $\lambda > 0$ and $0 < \gamma \le 1$.

It is generally agreed that a rational, loss-averse, risk-averse individual should have
a utility function which is increasing and concave. Thus Kahneman and Tversky's
result appears to give empirical evidence for the hypothesis that either individuals do not behave rationally or that they are not risk-averse.

Alternatively one might argue that the apparent irrationality is due to failing to fully analyse the actual returns experienced by actors. For example, a particular trader may be interested in their pay-packet and not in the performance of their portfolios. Thus the fact that a trader may be willing to risk enormous losses is perfectly rational: they personally only lose their job and possibly their reputation even if they bring down the bank they are working for. The ``true'' utility function of the trader should be a function of their pay-packet. By considering only functions of the PnL of the portfolios they manage, one is given a false impression that the traders are irrational. Similarly, it is perfectly rational for a limited liability company to take enormous risks with other people's money. 

Whether the cause of S-shaped utility functions is irrationality or limited liability, there is certainly good evidence that they are a useful tool for
modelling real worlds behaviour. A regulator or risk-manager should certainly
consider the possibility that they must regulate or manage actors who behave as though governed by S-shaped utility.

Not all of the characteristics of S-shaped utility functions are important to
us in this paper.  We are primarily interested in the convexity on the left. Motivated by Kahneman and Tversky's original example \eqref{eq:ktexample}, we introduce the following definition.

\begin{definition}
An increasing function $u:{\mathbb R}\longrightarrow{\mathbb R}$ (to be thought of as a utility function) is said to be {\em ``risk-seeking in the left tail''} if there exist constants $N \leq 0$, $\eta \in (0,1)$ and $c>0$ such that:
\begin{equation}
u(x) > -c|x|^\eta \qquad \forall x \leq N.
\label{tailDefinition}
\end{equation}

Similarly $u$ is said to be {\em ``risk-averse in the right tail''} if there exists $N \geq 0$, $\eta \in (0,1)$ and $c>0$ such that
\begin{equation}
u(x) < c|x|^\eta \qquad \forall x \geq N.
\label{rightTailDefinition}
\end{equation}
\end{definition}

The standard pictures of ``S-shaped'' utility functions in the literature appear to have these properties. Furthermore the S-shaped utility functions that would arise due to a limited liability would be bounded below and so would certainly be risk-seeking in the left tail.

We give a formal definition of S-shaped for the purposes of this paper.
\begin{definition}
	A function $u$ is said to be ``{\em S-shaped}'' if
	\begin{enumerate}
		\item  $u$ is increasing
		\item  $u(x)\leq 0$ for $x \leq 0$
		\item  $u(x)\geq 0$ for $x \geq 0$
		\item  For $x\geq 0$, $u(x)$ is concave.
		\item  $u$ is risk-seeking in the left tail.
		\item  $u$ is risk-averse in the right tail.
	\end{enumerate}
\end{definition}

\section{Law-invariant portfolio optimization}\label{sec:lawinvariant}

Let $(\Omega, {\cal F}, \P)$ be a probability space and let $\frac{\ed \Q}{\ed \P}$ be a positive random variable with $\int_{\Omega} \frac{ \ed \Q}{\ed \P} \, \ed \P(\omega) = 1$. 

We will use this model to represents a complete financial market as follows:
\begin{enumerate}[(i)]
\item We assume there is a fixed risk free interest rate $r$, assumed to be a deterministic constant. 
\item Given a measurable function, or random variable $f$, one can purchase a derivative security with payoff at time $T$ given by $f(\omega)$ for the price
\begin{equation}
\E^\Q[ e^{-rT} f] := \int_{\Omega} e^{-rT} f(\omega) \frac{\ed \Q}{\ed \P}(\omega) \, \ed \P
\label{eq:riskNeutralPrice}
\end{equation}
assuming that this integral exists.
\end{enumerate}

We note that the properties we require of $\frac{\ed \Q}{\ed \P}$ allow us to define a measure $\ed \Q:=\frac{\ed \Q}{\ed \P} \ed \P$, justifying our notation.

In convex analysis it is often convenient to allow infinite values in calculations (see \cite{rockafellar}). We will use the following conventions.
Let us write $f^+$ and $f^-$ for the positive and negative parts of a measurable function $f$. Suppose that
\[
\int_{\Omega} f^+(\omega) \, \ed \P
\]
is finite but
\[
\int_{\Omega} f^-(\omega) \, \ed \P
\]
is not finite, then we will write
\[
\int_{\Omega} f(\omega) \, \ed \P = -\infty.
\]
We can similarly define what it means for an integral to equal $+\infty$.

We will be considering investment under cost constraints. In our model we will assume that it is possible to purchase a derivative with payoff $f(\omega)$ whose price
according to \eqref{eq:riskNeutralPrice} is $-\infty$ whatever cost constraint is imposed. Assets where the cost is $+\infty$, or where the price is undefined, cannot be purchased. 

We assume that the investor's preferences are encoded by some function 
\[
v: \M \to \R
\]
where $\M$ is the space of probability measures on $\R$, 
so that an investor will prefer a security with payoff $f$ over a security with payoff $g$ iff $v(F_f)>v(F_g)$ (we are writing $F_f$ for the cumulative density function of the random variable $f$). Thus the investor's preferences are law-invariant.

We assume that the investor has a fixed budget $C$ so that they can only purchase securities with payoff $f$ satisfying
\[
-\infty \leq \E^\Q[ e^{-rT} f]=   \int_{\Omega} e^{-rT} f(\omega) \frac{\ed \Q}{\ed \P}(\omega) d\P(\omega) \leq C.
\]
We emphasize that if this integral is equal to $-\infty$ the investor may purchase the security.
As we hinted above, treating infinity in this way is notationally convenient and is standard practice in convex analysis \cite{rockafellar}.  At an intuitive level, we are simply saying that one can always purchase an asset if one is willing to overpay.

We assume that all the other trading constraints are law-invariant. For example:
the investor may have to ensure that the minimum payoff is almost surely above a certain level; they may be operating under $\text{ES}$ or $\text{VaR}$ constraints; they may be operating under utility constraints. We model the
combined constraints using a set ${\cal A}\subseteq{\cal M}_1(\R)$ and
requiring that $F_f \in {\cal A}$.

In summary, our investor wishes to solve the following optimization problem:
\begin{equation}
\begin{aligned}
& \underset{f \in L^0(\Omega,\P)}{\maximizeterm}
& & v(F_f) \\
& \text{subject to a price constraint}
& & -\infty \leq \int_{\Omega} e^{-rT} f(\omega) \frac{\ed \Q}{\ed \P}(\omega) \, \ed \P(\omega) \leq C \\
& \text{and risk management constraints}
& & F_f \in {\cal A} \subseteq {\cal M}_1(\R).
\end{aligned}
\label{optimizationProblem}
\end{equation}

In Appendix \ref{appendix:rearrangementProof} we prove the following theorem which shows we may assume $f$ has a particular form when solving
\eqref{optimizationProblem}.
Let $F^{-1}_f$ denote generalized inverse of the cumulative distribution function $F_f$ defined by $F^{-1}_f(p):=\inf\{x:F_f(x)\geq p\}$. We similarly write $(1-F_f)^{-1}$ for the generalized inverse of complementary cumulative distribution function which is defined by $(1-F_f)^{-1}(p):=\inf\{x \mid 1-F_f(x) \leq p \}$.

\begin{theorem}
\label{theorem:rearrangement}
Suppose that $\Omega$ is non-atomic then there exists a
uniformly distributed random variable $U$
such that:
\begin{enumerate}[(i)]
\item $\frac{\ed \Q}{\ed \P}=(1-F_\frac{\ed Q}{\ed \P})^{-1}\circ U$ almost surely.
\item If $f$ satisfies the price and risk management constraints of \eqref{optimizationProblem}
\ then \[ \varphi(U) =F_f^{-1} \circ U \] also satisfies the constraints of \eqref{optimizationProblem}
and is equal to $f$ in distribution, and hence has the same
objective value as $f$.
\end{enumerate}
\end{theorem}

\section{Portfolio optimization with S-shaped utility and expected shortfall constraints}\label{sec:optimiz}

Let $u$ be a function which need not necessarily be either concave or increasing. Consider problem \eqref{optimizationProblem} where the objective, $v$,  
is the expected utility for $u$ and where we have a single expected shortfall constraint. For a definition of value at risk and expected shortfall we refer for example to \cite{mcneil} or \cite{acerbitasche}. Suppose our probability model is non-atomic, and let $U$ be the random variable given in Theorem \ref{theorem:rearrangement} and define $q=1-F_\frac{\ed \Q}{\ed \P}$, so $q$ is a decreasing function.

By Theorem \ref{theorem:rearrangement}, under expected shortfall risk constraint the optimization problem is equivalent to
solving
\begin{align}
\underset{\varphi:[0,1]\to \R, \varphi \text{ increasing}}{\maximizeterm} &\quad {\cal F}(\varphi) := \quad \int_0^1 u(\varphi(x))\, \hbox{d}x  
\label{problem} \\
\text{subject to the price constraint} &\quad \int_0^1 \varphi(x) q(x) \, \hbox{d} x \leq C
\label{budgetConstraint} \\
\text{and the expected shortfall constraint} &\quad \frac{1}{p} \int_0^p \varphi(x) \, \hbox{d} x \geq L.
\label{ESConstraint}
\end{align}


Moreover, this map preserves the objective values and the supremum. Note that the expected shortfall representation in the left hand side of \eqref{ESConstraint} comes from (3.3) in \cite{acerbitasche}.

We are now ready to state the main negative result of this paper.

\begin{theorem}[Irrelevance of expected shortfall constraints under tail risk-seeking]
	\label{thm:ESNonBinding}
	Suppose $u$ is risk-seeking in the left tail and 
	\[ \lim_{x \rightarrow 0} q(x) = \infty \]
	then the  optimal value of the optimization problem
		\eqref{problem} \eqref{budgetConstraint} under the ES constraint \eqref{ESConstraint}
		is the same as the optimal value of the unconstrained problem, $\sup u$.
\end{theorem}
\begin{proof}

We consider functions $\varphi$ of the form:
\begin{equation}
\varphi(x) = \begin{cases}
k_1, & \text{if }x \geq \alpha \\
k_2, & \text{otherwise.}
\end{cases}
\end{equation}
We require that $0<\alpha<p$ and $k_2<k_1$.
For functions of this form we
can rewrite \eqref{problem} as:
\begin{equation}
{\cal F}(\varphi) = \alpha u(k_2) + (1-\alpha)u(k_1)
\label{objectiveFunction2}
\end{equation}
and equations \eqref{budgetConstraint} and \eqref{ESConstraint} as:
\begin{equation*}
\frac{ p L - (p-\alpha)k_1}{\alpha} \leq k_2 \leq
\frac{ C - k_1 \int_\alpha^1 q(x) \, \ed x}
{\int_0^\alpha q(x) \, \ed x. }
\end{equation*}
Let us restrict ourselves further to functions where:
\begin{equation*}
k_2 = \frac{ p L-(p-\alpha)k_1}{\alpha}.
\end{equation*}
So long as $k_1$ is sufficiently large, we will have $k_2<k_1$.
For such functions the ES constraint is automatically satisfied and the
budget constraint becomes:
\begin{equation}
p L - (p-\alpha)k_1 \leq \frac{\alpha}{\int_0^\alpha q(x)\,\ed x}
\left(
C - \int_{\alpha}^{1} q(x)\, \ed x
\right).
\label{budgetConstraint2}
\end{equation}
Taking the limit of the left hand side of \eqref{budgetConstraint2} as $\alpha \rightarrow 0$ we obtain
\[ p L - p k_1. \]
On the other hand the right hand side tends to zero as $\alpha \rightarrow 0$ because
of our assumptions on the function $q(x)$. For all sufficiently large $k_1$ we can
ensure that
\[ p L - p k_1 < -1 < 0. \]
So for sufficiently large $k_1$ and sufficiently small $\alpha$, the budget
constraint will hold.

With the chosen value for $k_2$, our objective function can be written:
\begin{equation*} 
\label{objectiveFunction3}
{\cal F}(\varphi) = \alpha u \left( \frac{p L - (p-\alpha)k_1}{\alpha} \right) + (1-\alpha)u(k_1)
\end{equation*}
Our constraints are now simply that $0<\alpha<\delta$ and $k_1>M$ for some 
values $\delta>0$ and $M>0$.

We wish to show that for sufficiently large $k_1$, the limit of \eqref{objectiveFunction3}
as $\alpha$ tends to zero is $u(k_1)$. Let us choose constants $c$, $N$ and $\eta$ as
given in \eqref{tailDefinition}. We may choose $k_1$ sufficiently large and $\alpha$
sufficiently small so that the following all hold:
\[ k_1 > \frac{2 M}{p}, \]
\[ \frac{p}{4} k_1 > p L, \]
\[ \alpha < \frac{p}{2}. \]
It follows that
\[ p L < \left( \frac{p}{2} - \alpha \right) k_1 \]
and hence
\[ p L < (p -\alpha)X_1 - \frac{p}{2} k_1 < (p - \alpha) k_1 - M \alpha. \]
Thus
\[ \frac{ p L - (p -\alpha) k_1}{\alpha} < - M \]
and we may conclude
\[ u\left( \frac{ p L - (p-\alpha)k_1}{\alpha} \right) > -c \left| \frac{L - (p-\alpha)k_1}{\alpha} \right|^\eta. \]
So for sufficiently large $k_1$ and sufficiently small $\alpha>0$ the objective
function is bounded below:
\begin{equation*}
\begin{split}
{\cal F}(\varphi) & \geq -\alpha c \left| \frac{ p L-(p-\alpha)k_1}{\alpha} \right|^\eta
+ (1-\alpha)u(k_1). \\
&= - c \alpha^{1-\eta} | p L - (p-\alpha)k_1|^\eta + (1-\alpha)u(k_1) \\
&\rightarrow u(k_1) \hbox{ as } \alpha \rightarrow 0
\end{split}
\end{equation*}
Hence the supremum of the objective function is bounded below by $(\sup_x u(x)) - \epsilon$ for any $\epsilon>0$.
On the other hand it is trivial that the objective function  is bounded above by $\sup_x u(x)$.
\end{proof}

We can see the type of optimal payoff from the proof. A sketch of the optimal payoff can be deduced from Figure \ref{fig:optimalpayoff}. 
\begin{figure}
\includegraphics[scale=0.5]{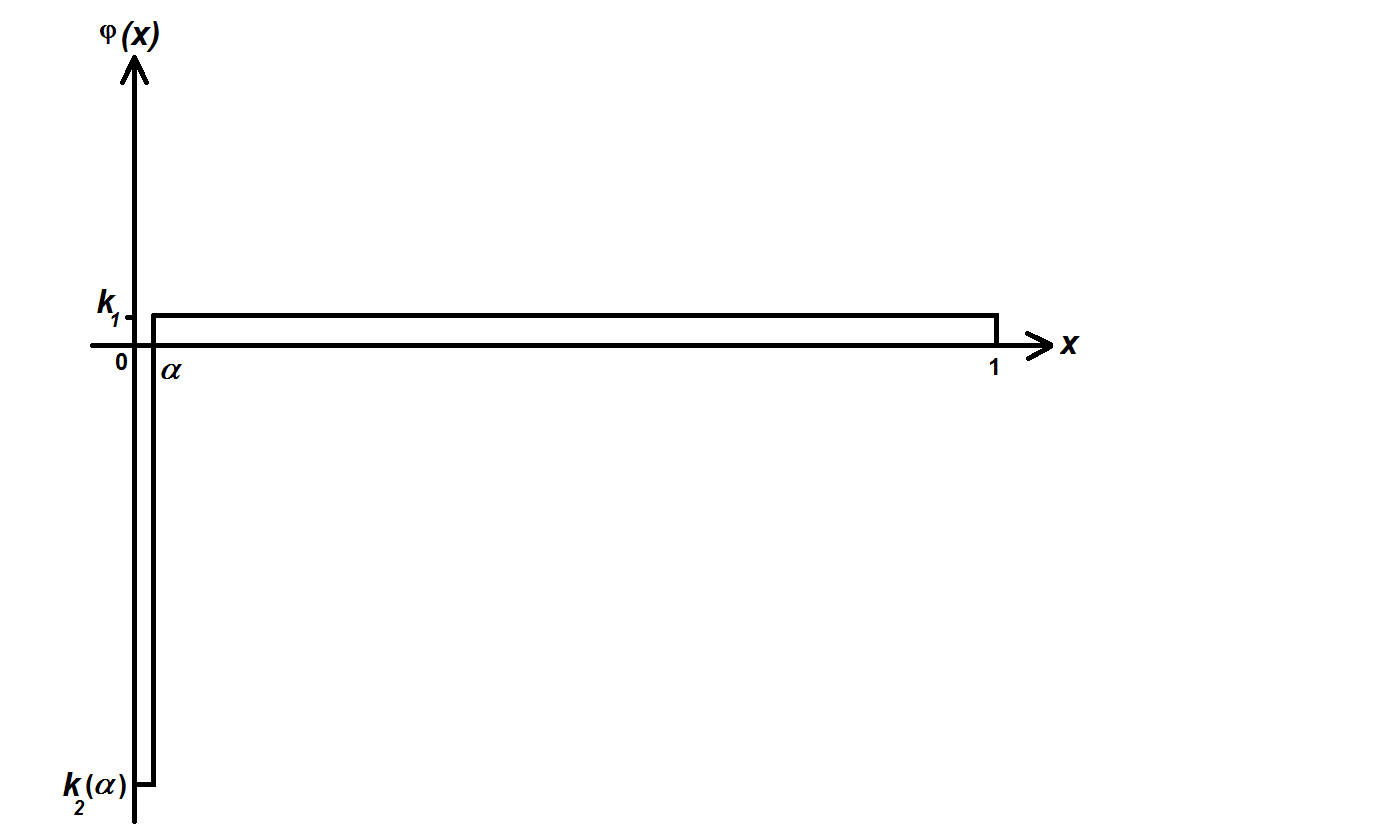}
\caption{Type of payoff whose limit for $\alpha \rightarrow 0$ (implying $k_2(\alpha)  \downarrow -\infty$) is used to find the optimal utility value}\label{fig:optimalpayoff}
\end{figure}
In this figure we focus on an example with positive $k_1$ and negative $k_2 = k_2(\alpha) = k_1 + p(L-k_1)/\alpha$. If we assume $L$ to be negative  then for $\alpha \downarrow 0$ we have $k_2 \downarrow - \infty$ for any positive $k_1$. Essentially the payoff we use is a digital option with a very large (in absolute value) negative value $k_2$ in a very small area of size $\alpha$ near $0$, and with a much smaller positive payoff $k_1$ in the big area $[\alpha,1]$. This is the type of payoff that satisfies the budget and expected shortfall (or VaR in case) constraints while maximizing the expected S-shaped utility. 


\section{Trading simple claims in the Black--Scholes--Merton market}
\label{section:blackScholesDiscrete}

In this section we will consider a investor who wishes to optimize her utility at time $T$ by investing in options with maturity $T$. This investor will follow a buy and hold strategy, but her portfolio will be an option portfolio.

We assume that put and call options can be bought and sold at a wide variety of strikes and so the market can be reasonably well approximated by a complete market. In this market any European derivative whose payoff is a function of the final stock price,  otherwise known as simple contingent claim, can be bought or sold at a fixed price.

We must choose a model for the price of these derivatives. As an
example we will consider European derivatives on a stock which follows the Black--Scholes--Merton model. Thus we consider derivatives on the final stock price in a market where one can trade in either zero coupon bonds with (deterministic) risk free rate $r$ or in a non-dividend paying stock whose price at time $t$, $S_t$, follows a geometric Brownian motion under the $\P$-measure
\[
\ed S_t = S_t( \mu \, \ed t + \sigma \, \ed W^\P_t) , \ \ S_0
\]
with drift $\mu$, volatility $\sigma>0$ and initial condition $S_0>0$. The process $W^\P_t$ is a standard Brownian motion under the $\P$ probability measure.

By Ito's formula, the log of the stock price, $s_t = \ln S_t$, satisfies
\[
\ed s_t = \left(\mu - \frac{1}{2}\sigma^2\right) \, \ed t + \sigma \, \ed W^{\P}_t, \ \  s_0 = \ln(S_0). 
\]
Hence, under the $\P$-measure, $s_T$ is normally distributed with mean $s_0 + (\mu - \frac{1}{2}\sigma^2)T$ and standard deviation $\sigma \sqrt{T}$. Let us
write the density function explicitly:
\begin{equation}
p_{s_T}^{\text{BS}}(x) = \frac{1}{ \sigma \sqrt{ 2 \pi T}} \exp\left( - \frac{(x-(s_0+ (\mu-\frac{1}{2}\sigma^2)T))^2}{2 \sigma^2 T} \right).
\label{eq:bsPDensity}
\end{equation}

The standard pricing theory in the Black--Scholes--Merton market tells us
that the price of a European derivative in this market can be computed using the discounted $\Q$-measure expectation of the payoff where the stock price process in the $\Q$-measure is
\[
\ed S_t = S_t(r \, \ed t + \sigma \, \ed W^\Q_t ), \ \ S_0,
\]
where now $W^\Q_t$ is a standard Brownian motion under the measure  $\Q$. 

Hence we can write down the $\Q$-measure density function for $s_T$
\begin{equation}
q^{\text{BS}}_{s_T}(x) = \frac{1}{ \sigma \sqrt{ 2 \pi T}} \exp\left( - \frac{(x-(s_0+ (r-\frac{1}{2}\sigma^2)T))^2}{2 \sigma^2 T} \right).
\label{eq:bsQDensity}
\end{equation}

Now let $(\Omega^D, {\cal F}^D, \P^D)$ be the probability space for
the final log stock price $s_T$. The superscript $D$ stands for ``derivatives market''. This probability space is simply $\R$
with probability density given by $p^{\text{BS}}$ and $\sigma$-field given by the Borel set.
We can define a random variable $\frac{\ed \Q}{\ed \P}^D(s_T) := q^\text{{BS}}(s_T)/p^\text{{BS}}(s_T)$ on $(\Omega^D, {\cal F}^D, \P^D)$. Together, $\frac{\ed \Q}{\ed \P}^D$ and $(\Omega^D, {\cal F}^D, \P^D)$ 
define a market: the market of European derivatives on the final stock price $s_T$. For any payoff function $f$ of $s_T$ the price of the derivative is given by \eqref{eq:riskNeutralPrice} so long as this integral exists and is less than $\infty$. This is a complete market.

Let $U$ be the standard uniform random variable given by $F_{s_T}(s_T)$. We calculate 
\[
F^{-1}_{s_T}(U) =s_0 + T \left(\mu -\frac{\sigma ^2}{2}\right) + \sigma  \sqrt{T} \Phi^{-1}(U)
\]
where as usual $\Phi$ us the cumulative distribution function of the standard normal. Using the explicit formulae for $q(s_T)$ and $p(s_T)$ we compute

\[
\begin{split}
\frac{\ed \Q}{\ed \P}^D(U) =
\frac{q(F^{-1}_{s_T}(U))}{p(F^{-1}_{s_T}(U))} 
&= 
 \exp \left(\frac{\mu-r}{2 \sigma^2} \left( (\mu+r-\sigma^2)T + 2 (s_0 - F^{-1}_{s_T}(U)) \right)\right)   \\
&=  \exp \left[\frac{\mu-r}{ \sigma} \left( -\frac{\mu-r}{\sigma}\frac{ T}{2} - \sqrt{T} \Phi^{-1}(U)\right)\right]
\end{split}
\]
where we have highlighted the role of the market price of risk $\frac{\mu-r}{\sigma}$.
If we assume $\mu > r$, then this function is decreasing in $U$. Since $U$ is uniform, we conclude that this expression is equal to $(1-F_{ \frac{\ed \Q}{\ed \P}}(U))^{-1}$. We see that if $\mu > r$ then $\frac{\ed \Q}{\ed \P}^D \rightarrow \infty$ as  $U \to 0$. If $\mu < r$ then $\frac{\ed \Q}{\ed \P}^D \rightarrow \infty$ as  $U \to 1$.

Thus European derivatives at time $T$ in the Black--Scholes--Merton market satisfy the assumptions of Theorem \ref{thm:ESNonBinding}. 
We have the following
\begin{corollary}[Irrelevance of expected shortfall in reducing tail-risk-seeking behaviour in a Black--Scholes market]
In a Black--Scholes model, the expected utility of a investor who is risk-seeking in the left tail is limited only by the supremum of her  utility function. Investors can achieve any desired  expected utility below this supremum by trading in the bond and a digital option. Expected shortfall constraints do not impact the expected utility corresponding to the optimal solution.
\end{corollary}
\begin{remark}
We have assumed in our analysis that the horizon of the investment, $T_1$,
and the ES time horizon, $T_2$ coincide. Typically the ES time horizon is the time one estimates could be needed to liquidate the position in a hostile market
so in practice one would have $T_2<T_1$. However, an investor who wishes to maximize their utility at time $T_1$ could choose to restrict themselves to buying derivatives with maturity $T_2$ and then holding investing the payoff in zero coupon bonds until time $T_1$ in which case the payoff at time $T_1$ would be a function of the payoff at time $T_2$.
\end{remark}
\begin{remark}
We have illustrated our results with the Black--Scholes--Merton model for simplicity. The key observation was that densities for $p^{BS}$ and $q^{BS}$
were normal but with different drifts, so the ratio $q^{BS}/p^{BS}$ is unbounded.  Over sufficiently short time horizons, the density of any stochastic process driven by It\^o equations can be well approximated by
multivariate normal distributions. This can be expressed rigorously using asymptotic formulae for the heat kernel of a stochastic process (see for example \cite{hsu}). Thus over short time horizons one expects to find that $\frac{\ed \Q}{\ed \P}$ will be unbounded for any market model defined using It\^o calculus where the market price of risk is non-zero. One can easily devise examples of stochastic processes which converge to a fixed value at a future time $T$, so we cannot deduce that $\frac{\ed \Q}{\ed \P}$ will be unbounded at time $T$. Nevertheless, one expects that in a realistic market model $\frac{\ed \Q}{\ed \P}$ will indeed be unbounded at any time $T$.
\end{remark}

\section{Portfolio optimization with limited liability and utility constraints}\label{sec:limlam}

We saw previously important but negative results: tail-risk-seeking investors and investors who aim to maximize S-shaped utilities are not impacted by expected shortfall constraints. While this tells us that in this context expected shortfall is not effective in curbing excessive risk taking, what should one do instead? We have reached the point in the paper where we can make a positive proposal for an alternative approach.

Let us return to problem \eqref{optimizationProblem}. We suppose that 
the regulator is indifferent to the outcome if the portfolio payoff is positive and imposes
a risk constraint on the expected utility of the negative part of the payoff. We  specialise our analysis to the case where the investor has limited liability and so is indifferent to the outcome if the portfolio payoff is negative. We model this by choosing two utility functions $u_R$ and $u_I$ representing
the regulator and the investor's utility functions respectively. 

\begin{theorem}
\label{theorem:limitedLiability}
Let $u_R(x):\R\to\R$ be a concave increasing function (associated with the risk constraint utility function) equal to $0$ when $x\geq0$.
Let $u_I(x):\R\to\R$ be an increasing function equal to $0$ when $x\leq 0$ and concave
in the region $x\geq 0$ (associated with the investor utility function).  Define $v:{\cal M}_1(\R) \to \R \cup \{\infty\}$
by the expected $u_I$-utility over a random variable distributed as $f$:
\[
v(F_f)=\int_{ - \infty}^{\infty} u_I(x) \, \ed F_f(x).
\]
Let $L$ be a negative real number. Define ${\cal A} \subseteq {\cal M}_1(\R)$ by
\[
{\cal A}=\left\{ F_f \mid \int_{ - \infty}^{\infty} u_R(x) \, \ed F_f (x) \geq L \right\}.
\]
This is the set of distribution functions leading to an expected $u_R$-utility larger than a possible ``loss" level $L$. 
It is worth mentioning that both the objective and the admissible set are formulated in terms of expected utility functions.

The supremum for the optimization problem \eqref{optimizationProblem} can then be computed
as follows.

Define $q(x)=(1-F_\frac{\ed \Q}{\ed \P})^{-1}(x)$. 

Given $p \in [0,1]$, define $C_1(p) \in \R \cup \{-\infty\}$ to be the infimum of the optimization problem
\begin{equation}
\begin{aligned}
\underset{f_1:[0,p]\to \  (-\infty,0), \text{ with $f_1$ increasing}
}{\inf} &  \quad \int_0^p f_1(x) q(x) \ed x \\
\text{subject to} & \quad \int_0^p u_R(f_1(x)) \, \ed x \geq L. \\
\end{aligned}
\label{leftOptimizationProblem}
\end{equation}

Define $V(p) \in \R \cup \{\infty\}$ to
be the supremum of the optimization problem
\begin{equation}
\begin{aligned}
\underset{f_2:[p,1]\to [0,\infty), \text{ with $f_2$ increasing}
}{\maximizeterm} &  \quad \int_p^1 u_I(f_2(x)) \ed x \\
\text{subject to} & \quad \int_p^1 f_2(x) q(x) \ed x \leq C_2. \\
\end{aligned}
\label{rightOptimizationProblem}
\end{equation}
with $C_2(p)$ defined by
\begin{equation}
C_2(p) := e^{r T} C - C_1(p)
\label{eq:c2Definition}
\end{equation}
(recall that $C_1$ comes from the first problem above). 
Then the supremum of the problem \eqref{optimizationProblem} is equal to
$\sup_{p\in[0,1]} V(p)$.
\end{theorem}
\begin{remark}
The value of the theorem comes from
the fact that the problems \eqref{leftOptimizationProblem} and
\eqref{rightOptimizationProblem} are easy to solve, see Lemma \ref{lemma:convexOptimization} below. One
may then compute $\sup_{p\in[0,1]} V(p)$
by line search. Moreover, it is simple
to obtain an explicit solution of \eqref{optimizationProblem} given solutions to each of these simpler problems. The risk constraint will typically be relevant, unlike the case of expected shortfall.
\end{remark}

\begin{remark}
Although we have specialised to the case of limited liability, note that the strategies pursued by an investor with limited liability will be at least as aggressive than those pursued by an investor with S-shaped utility. Thus if we can find bounds for an investor with limited liability, we will
obtain bounds for more general S-shaped utilities. Finding explicit solutions to the problem \eqref{optimizationProblem} for general S-shaped utilities would seem a rather more difficult problem.
\end{remark}

\begin{proof}
By Theorem \ref{theorem:rearrangement}, the optimization problem \ref{optimizationProblem} is equivalent to solving
\begin{equation}
\begin{aligned}
\underset{\varphi:[0,1]\to \R, \text{ with $\varphi$ increasing}}{\maximizeterm} & \quad \int_0^1 u_I(\varphi(x)) \, \ed x \\
\text{subject to} & \quad \int_0^1 u_R(\varphi(x)) \, \ed x \geq L \\
\text{and} & \quad  -\infty \leq \int_0^1 \varphi(x) q(x) \, \ed x \leq e^{rT} C.
\end{aligned}
\label{optimizeWithUtilityConstraint}
\end{equation}
where $q=(1-F_\frac{\ed \Q}{\ed \P})^{-1}$ and so is decreasing with integral $1$.

Since in problem \eqref{optimizeWithUtilityConstraint} we
require that $\varphi$ is increasing, there is some $p\in[0,1]$ such that $\varphi(x)$ is less than $0$ for $x$ less than $p$ and $\varphi(x)$ is greater than $0$ for $x$ greater than $p$. Since the value of the integrals in the optimization problems is unaffected by the value of $\varphi$ at the single point $p$, we may also assume that $\varphi(p)=0$.

For a fixed $p$, we may define $f_1$ to be the restriction of $\varphi$ to $[0,p]$ and $f_2$ to be the restriction of $\varphi$ to $[p,1]$. Let us write
$\tilde{V}(p)\in \R \cup \{\pm \infty\}$ for the value of the supremum in the problem. 
\[
\begin{aligned}
\underset{f_2:[p,1]\to [0,\infty), \text{ with $f_2$ increasing}
}{\underset{f_1:[0,p]\to [-\infty,0), \text{ with $f_1$ increasing}}{\maximizeterm}} & 
\quad \int_p^1 u_I(f_2(x)) \, \ed x \\
\text{subject to} & \quad \int_0^p u_R(f_1(x)) \, \ed x \geq L \\
\text{and} & \quad  - \infty \leq \int_0^p f_1(x) q(x) + \int_p^1 f_2(x) q(x) \, \ed x \leq e^{rT} C.
\end{aligned}
\]

We use the value $-\infty$ to indicate that the constraints cannot be satisfied
as is conventional in convex analysis.
The supremum of problem \eqref{optimizeWithUtilityConstraint} and hence of \eqref{optimizationProblem} is given by
\[
\underset{p\in[0,1]}{\sup} \tilde{V}(p).
\]
It is obvious that $V(p)=\tilde{V}(p)$.
\end{proof}

	The case when the supremum of the optimization problem \eqref{rightOptimizationProblem} is equal to the supremum of the investor's utility function $u_I$ is rather uninteresting as the risk-constraints
	clearly will play no role. This motivates the following definition.
	\begin{definition}
		In a given market, an investor with utility function $u_I$ is said to be {\em difficult to satisfy} if the supremum of the optimization problem \eqref{rightOptimizationProblem} is less then the supremum of their utility function for any finite cost constraint $C_2$ and any $p \in(0,1)$.
	\end{definition}

\begin{lemma}
\label{lemma:convexOptimization}
Let $A$ and $B$ be constants satisfying $-\infty\leq A < B \leq \infty$
and let $a$ and $b$ be finite constants satisfying $a<b$. We will write
$\mathcal I$ for the set of increasing functions mapping $[a,b]$ to $[A,B]$.

Suppose that $q:[a,b]\to \R$ is a positive decreasing function with finite integral. Suppose that $u$ is a concave increasing function. Let $\partial u(x)$ be the set
\[
\partial u (x) = \{ y\in[0, \infty) \mid \forall x^\prime \in [A,B],  u(x^\prime) \leq u(x) + y(x^\prime - x)
 \}.
\]
Apart from at the boundary points $\{A,B\}$, this is the subdifferential of the concave function $u$\footnote{See \cite{rockafellar} for a discussion of subdifferentials. As remarked in \cite{rockafellar} the term subdifferential is used for both convex and concave functions even though superdifferential might be considered a more apt term for concave functions.}.
Let $\alpha$ be a constant and let $\varphi^\star \in F$ satisfy
\begin{equation}
 \alpha q(x) \in  \partial u(\varphi^*(x)) 
\label{eq:optimalityCondition}
\end{equation}
for every $x$. Then $\varphi^*$ is a solution to the maximization problem
\begin{equation}
\begin{aligned}
\underset{\varphi \in F}{\maximizeterm} & \quad \int_a^b u(\varphi(x)) \, \ed x \\
\text{subject to} & \int_a^b \varphi(x) q(x) \, \ed x \leq \int_a^b \varphi^*(x) q(x) \, \ed x
\end{aligned}
\label{eq:boundedMaximization}
\end{equation}
and the minimization problem
\begin{equation}
\begin{aligned}
\underset{\varphi \in F}{\inf} & \quad \int_a^b \varphi(x) q(x) \, \ed x \\
\text{subject to} & \int_a^b u(\varphi(x)) \, \ed x \geq \int_a^b u(\varphi^*(x)) \, \ed x.
\end{aligned}
\label{eq:boundedMinimization}
\end{equation}
\end{lemma}
\begin{proof}
Let $\varphi:[a,b]\to [A,B]$ be another function. By the assumption \eqref{eq:optimalityCondition} we have
\[
u(\varphi(x)) \leq u(\varphi^*(x)) + \alpha q(x) (\varphi(x)-\varphi^*(x)).
\]
Integrating this
\[
\int_a^b u(\varphi(x)) \, \ed x \leq \int_a^b u(\varphi^*(x)) \, \ed x + \alpha \int_a^b q(x) (\varphi(x)-\varphi^*(x)) \, \ed x.
\]
So if $\varphi$ satisfies the constraints of \eqref{eq:boundedMaximization} we conclude
\[
\int_a^b u(\varphi(x)) \ed x \leq \int_a^b u(\varphi^*(x)) \, \ed x.
\]
Thus $\varphi^*$ solves the problem \eqref{eq:boundedMaximization}. Similarly $\varphi^*$ solves \eqref{eq:boundedMinimization}.
\end{proof}
\begin{remark}
Lemma \ref{lemma:convexOptimization} can also be used to solve portfolio optimization problems
of the form \eqref{optimizationProblem} where the only constraints are
bounds on the payoff function $f$. These problems are considered in more
detail in \cite{follmerBook}, with a greater emphasis on the uniqueness of the solutions.
\end{remark}

Continuing with the propositive part of the paper, we now compute the solution of the
problem in Theorem \ref{theorem:limitedLiability} in a
specific case. The significance of this computation
is that it will allow us to immediately
write down an upper bound on the solution of the problem in \ref{theorem:limitedLiability} for
many financially interesting cases.

\begin{theorem}
Let $\gamma_R \in (1, \infty)$ be given. Let
\[
u_R(x) = \begin{cases}
-(-x)^{\gamma_R} & x \leq 0 \\
0 & \text{otherwise}. \\
\end{cases}
\]
Suppose that we wish to solve the optimization problem of Theorem \ref{theorem:limitedLiability} and that $u_I$ is such that the investor is difficult to satisfy. The risk constraint in Theorem \ref{theorem:limitedLiability} is binding if and only if the expectation
\[
\E_\P \left( \frac{ \ed \Q}{\ed \P}^\frac{\gamma_R}{\gamma_R-1} \right)
\]
is finite.

If the investor's utility function is given by
\[
u_I(x) = \begin{cases}
x^{\gamma_I} & x \geq 0 \\
0 & \text{otherwise} \\
\end{cases}.
\]
for $\gamma_I \in (0,1)$, then the investor is difficult to satisfy if the expectation
\[
\E_\P \left( \frac{ \ed \Q}{\ed \P}^\frac{\gamma_I}{\gamma_I-1}\textsl{} \right)
\]
is finite.
\label{thm:gammaUtility}
\end{theorem}
\begin{proof}
In this case $u_R$ is smooth with derivative
\[
u_R^\prime(x) = \gamma_R(-x)^{\gamma_R-1}.
\]
We define $i_1(y)=((u_R)^\prime)^{-1}(y):[0,\infty)\to (-\infty,0]$.  So
\[
i_1(y)=-\left(\frac{y}{\gamma_R}\right)^\frac{1}{\gamma_R-1}.
\]
Given $\alpha>0$ we define $\varphi^*_{1,\alpha}(x)=i_1(\alpha q(x))$.  By Lemma \ref{lemma:convexOptimization}, $\varphi^*_{1,\alpha}$ is a solution of the problem:
\begin{equation}
\begin{aligned}
\underset{\varphi \in F}{\inf} & \quad \int_a^b \varphi(x) q(x) \, \ed x \\
\text{subject to} & \int_a^b u_R(\varphi(x)) \, \ed x \geq L(\alpha,a,b).
\end{aligned}
\label{abOptimization}
\end{equation}
where $0\leq a<b \leq 1$ and
\[
\begin{split}
L(\alpha,a,b) &:= \int_a^b u_R(\varphi_{1,\alpha}^*(x)) \ed x \\
&= -\int_a^b \left( \left( \frac{\alpha q(x)}{\gamma_R} \right)^{\frac{1}{\gamma_R-1}} \right)^{\gamma_R} \, \ed x \\
&= -\left( \frac{\alpha}{\gamma_R} \right)^\frac{\gamma_R}{\gamma_R-1}\int_a^b
q(x)^\frac{\gamma_R}{\gamma_R-1} \, \ed x. \\
\end{split}
\]
The optimum value
of \eqref{abOptimization} is given by
\[
\begin{split}
C_1(\alpha,a,b)
&:=\int_a^b q(x) i_1(\alpha q(x))\,\ed x \\
&=\int_a^b -q(x) \left( \frac{\alpha q(x)}{\gamma_R} \right)^\frac{1}{-1 + \gamma_R} \,\ed x \\
&=-\left( \frac{\alpha}{\gamma_R} \right)^{\frac{1}{\gamma_R-1}} \int_a^b q(x)^{\frac{\gamma_R}{\gamma_R-1}} \,\ed x.
\end{split}
\]

Let us write
\[
I_1(a,b):=\int_a^b q(x)^\frac{\gamma_R}{\gamma_R-1} \, ed x
\]
To solve ensure $L(\alpha,a,b)=L$ we must take
as $\alpha$
\[
\alpha^* = \gamma \left( \frac{-L}{I_1(a,b)}
\right)^\frac{\gamma_R-1}{\gamma_R}
\]
which we note is finite and is greater than $0$ whenever $I_1(a,b)$ is finite.

We compute that 
\[
\begin{split}
C_1(\alpha^*,a,b)&= - \left(
\left( \frac{-L}{I_1(a,b)}
\right)^\frac{\gamma_R-1}{\gamma_R}
\right)^{\frac{1}{\gamma_R-1}} I_1(a,b) \\
&= - (-L)^{\frac{1}{\gamma_R}} I_1(a,b)^\frac{\gamma_R-1}{\gamma_R}.
\end{split}
\]

If $I_1(0,p)$ is finite it follows
from Lemma \ref{lemma:convexOptimization}
that the $C_1$ of Theorem \ref{theorem:limitedLiability} takes the value
\begin{equation}
C_1(p) = -(-L)^{\frac{1}{\gamma_R}} I_1(0,p)^\frac{\gamma_R-1}{\gamma_R}.
\label{eq:c1Bound}
\end{equation}
If the investor is difficult to satisfy, it follows that the constraint is binding.

On the other hand if $I_1(0,p)$ is infinite,
we may take $\varphi_1(x)=\varphi^*_{\alpha^*}(x) {\mathbbm 1}_{[a,b]}(x)$ to find a function satisfying the constraints of problem
\eqref{leftOptimizationProblem} with
objective value $C(\alpha^*,a,b)$. Since
this tends to $-\infty$ as $a\to0$ we deduce
that
\[
C_1(p)=-\infty
\]
if $I_1(0,p)$ is infinite. Since $u_I$ is increasing we can achieve arbitrary large utilities below $\sup u_I$ given sufficient cash. Hence the constraint is not binding in this case.

To determine when the investor is easily satisfied we solve the optimization problem \eqref{rightOptimizationProblem}. We define
\[
i_2(y):=((u_I)^\prime)^{-1}(y)=\left(\frac{y}{\gamma_I}\right)^\frac{1}{\gamma_I-1}.
\]
We now define $\varphi_{2,\alpha}^*=i_2(\alpha(q(x))$ for $\alpha>0$.

 By Lemma \ref{lemma:convexOptimization}, $\varphi^*_{2,\alpha}$ is a solution of the problem:
\begin{equation}
\begin{aligned}
\underset{\varphi \in F}{\sup} & \quad \int_a^b u_I( \varphi(x)) \, \ed x \\
\text{subject to} & \int_a^b \varphi(x) q(x) \, \ed x \leq C_2(\alpha,a,b).
\end{aligned}
\label{abOptimizationRight}
\end{equation}
where $0\leq a<b \leq 1$ and
\[
\begin{split}
C_2(\alpha,a,b) &:= \int_a^b \varphi_{2,\alpha}^*(x) q(x) \ed x \\
&= \int_a^b \left(\frac{\alpha q(x)}{\gamma_I}\right)^\frac{1}{\gamma_I-1} q(x) \ed x \\
&= \int_a^b \left(\frac{\alpha}{\gamma_I}\right)^\frac{1}{\gamma_I-1} q(x)^\frac{\gamma_I}{\gamma_I-1} \ed x. \\
\end{split}
\]
We define $I_2(a,b)=\int_a^b q(x)^\frac{\gamma_I}{\gamma_I-1} \ed x$ so we have
\begin{equation}
C_2(\alpha,a,b) = \left(\frac{\alpha}{\gamma_I}\right)^\frac{1}{\gamma_I-1} I_2(a,b).
\end{equation}
The corresponding supremum of \eqref{abOptimizationRight} is then given by
\begin{equation}
\begin{split}
u(\alpha,a,b)&:=
\int_a^b \left( \frac{\alpha q(x)}{\gamma_I} \right)^\frac{\gamma_I}{\gamma_I-1} \ed x \\
&=
\left( \frac{\alpha}{\gamma_I}\right)^\frac{\gamma_I}{\gamma_I-1} I_2(a,b) \\
&= \left(\frac{C_2(\alpha,a,b)}{I_2(a,b)}\right)^{\gamma_I} I_2(a,b) \\
&= C_2(\alpha,a,b)^{\gamma_I} I_2(a,b)^{1-\gamma_I} \\
\end{split}
\label{eq:utilityBound}
\end{equation}

We deduce that the investor is difficult to satisfy if and only if
\[
I_2(a,b)
\]
is finite.
\end{proof}

We now summarize the key findings of the paper in a single theorem.

\begin{theorem}
\label{thm:summary}
Suppose that $(\Omega,{\cal F}, \P)$ and
$\frac{\ed \Q}{\ed \P}$ define a complete market. Define
\begin{equation}
e(\gamma):=\E_\P\left( \frac{\ed \Q}{\ed \P}^\frac{\gamma}{\gamma-1} \right).
\label{eq:integralQGamma}
\end{equation}
An investor with S-shaped utility function $u_I$ 
who is subject only to expected shortfall constraints can find a sequence of portfolios satisfying these constraints whose expected $u_I$-utility tends to infinity. If in addition the investor is difficult to satisfy and $u_R$ is the function
\begin{equation}
u_R = \begin{cases}
-(-x)^{\gamma_R} & x \leq 0 \\
0 & \text{otherwise} \\
\end{cases}
\label{eq:weakUtilityFunction}
\end{equation}
with $\gamma_R>1$ and $e(\gamma_R)$ finite, then any sequence of portfolios whose expected $u_I$-utility tends to infinity will have
expected $u_R$ utility tending to $-\infty$.

The conditions that $e(\gamma_R)$ is finite and that the investor is difficult to satisfy
are always satisfied for the market of derivatives on the Black--Scholes--Merton market described in Section
\ref{section:blackScholesDiscrete} when the market price of risk is non-zero
and the investor utility function $u_I$ is not bounded above.
\end{theorem}
\begin{proof}
Apart from the assertions about the Black--Scholes--Merton market, the result follows from Theorem \ref{theorem:limitedLiability},
Theorem \ref{thm:ESNonBinding} and our formal definition of S-shaped utility.
	
In the Black--Scholes--Merton market, the expectation $e(\gamma)$ is equal to
\[
\int_{\R} \left( \frac{q^{BS}_{s_T}(x)}{p^{BS}_{s_T}(x) }\right)^\frac{\gamma}{\gamma-1} p^{BS}_{s_T}(x) \, \ed x
\]
where $p^{BS}$ and $q^{BS}$ are given by equations \eqref{eq:bsPDensity} and
\eqref{eq:bsQDensity} respectively. On substituting in these formulae for $p^{BS}$ and $q^{BS}$ one obtains
\begin{equation}
e(\gamma) = \int_\R \frac{1}{\sigma  \sqrt{2 \pi  T}}\exp \left(\frac{c_0+c_1 x+c_2 x^2}{8 (\gamma -1) \sigma ^2 T}\right)\, \ed x 
\label{eq:eGammaExplicit}
\end{equation}
with
\[
\begin{split}
c_0 &= T^2 \left(-\gamma  \sigma ^4+4 \mu ^2-4 \mu  \sigma ^2-4 \gamma  r^2+4 \gamma  r \sigma ^2+\sigma ^4\right) \\
&\quad {}-4 s_0 T \left(-\gamma  \sigma ^2-2 \mu +2 \gamma  r+\sigma ^2\right)-4 (\gamma -1) s_0^2, \\
c_1 &= 4 \left(T \left(-\gamma  \sigma ^2-2 \mu +2 \gamma  r+\sigma ^2\right)+2 (\gamma -1) s_0\right), \\
c_2 &= 4-4 \gamma.
\end{split}
\]
The overall coefficient of $x^2$ in the exponential in our expression \eqref{eq:eGammaExplicit} for $e(\gamma)$ is
\[
\frac{4-4\gamma}{8 (\gamma-1) \sigma^2 T} = -\frac{1}{2 \sigma^2 T}.
\]
This is always negative and so the expression \eqref{eq:eGammaExplicit} is a Gaussian integral and hence is finite.

It now follows from Theorem \ref{thm:gammaUtility} that since $u_I$ is risk-averse on the right and also unbounded on the right that the investor is difficult to satisfy.
\end{proof}

\section{Conclusions}

We have shown that in typical complete markets with non-zero market price of risk, expected shortfall constraints do not affect the supremum of the investor utility that can be achieved by an investor with S-shaped utility, $u_I$. By contrast, even very weak expected utility constraints for a concave increasing limiting utility function can reduce the supremum that can be achieved.
In these circumstances, if a risk manager with such a concave increasing utility function, $u_R$, only imposes expected shortfall constraints, they should expect that a rogue investor will choose investment strategies with infinitely bad $u_R$-utilities.  These findings were stated in full detail in Theorem \ref{thm:summary}.
	
We believe that this shows that in complete markets value at risk or expected shortfall constraints alone are insufficient to constrain the behaviour of rogue investors.

An obvious criticism of our approach is that the complete market assumption is unrealistic. In many situations markets can be well approximated by complete markets, but they are a mathematical idealisation. In particular, our result that the strategies pursued are {\em infinitely} bad when measured using expected utilities will clearly fail in realistic markets. One expects that in practice the infinitely bad expected $u_R$-utilities will merely be very bad utilities. We will investigate this question numerically in future research.

Nevertheless, even if one believes that further research will reveal some features of realistic market models that significantly blunt our findings, it surely behoves those risk-managers who are willing to rely on expected shortfall constraints to explain what these features are, and to demonstrate the effectiveness of their risk-management constraints. 

\bibliography{sshaped}
\bibliographystyle{plain}

\newpage

\section*{Appendix}

\appendix

\section{Proof of Theorem \ref{theorem:rearrangement}}
\label{appendix:rearrangementProof}

 To prove Theorem \ref{theorem:rearrangement}, we define the notion of $X$-rearrangement.

\begin{definition}
	Given random variables $X, f \in L^0(\Omega, \R)$ with $X$ having a continuous
	distribution we define the {\em $X$-rearrangement of $f$}, denoted $f^X$ by:
	\[
	f^X(\omega) = F_f^{-1}( \P(X\leq X(\omega) ) ) = F_f^{-1}( F_X(X(\omega))) .
	\]
\end{definition}

\begin{lemma}
	\label{lemma:basicPropertiesOfRearrangement}
	The $X$-rearrangement has the following properties:
	\begin{enumerate}[(i)]
		\item If $X$ has a continuous probability distribution then $f^X$ is equal to $f$ in distribution.
		\item If $k \in \R$ then $(\max\{f,k\})^X = \max\{f^X, k\}$ and
		$(\min\{f,k\})^X = \min\{f^X, k\}$
		\item $f^X = (f^+)^X + (f^-)^X$.
		\item $X^X = X$ almost surely.
		\item If $g(\omega)=G(X(\omega))$ with $G$ increasing and if $X$ has a continuous probability distribution then $g^X=g$ almost surely. 
	\end{enumerate}
\end{lemma}
\begin{proof}[Proof of (i)]
	We recall that: for any distribution function $F$ with generalized inverse $F^{-1}$,  $F^{-1}(p)\leq x$ if and only if $p \leq F ( x )$;    
	$F_X\circ F^{-1}_X=\text{id}$
	if $X$ has a continuous distribution. Hence if $X$ has a continuous distribution:
	\[
	\begin{split}
	F_{f^X}(y) &= \P(f^X(\omega)\leq y) \\
	&= \P( F^{-1}_f( \P( X \leq X(\omega) )\leq y) \\
	&= \P( \P(X \leq X(\omega) ) \leq F_f(y) ) \\
	&= \P( F_X(X(\omega) )\leq F_f(y) ) \\
	&= \P( X(\omega) \leq F^{-1}_X F_f(y) ) \\
	&= F_X(F^{-1}_X(F_f(y))) \\
	& = F_f(y). \\
	\end{split}
	\]
\end{proof}
\begin{proof}[Proof of (ii)] 
	The result follows from the definition of $f^X$ and the following identities:
	\[
	\begin{split}
	F^{-1}_{\max\{f,k\}}(t)&=\inf\{z \in \R: \P(\max\{f,k\} \leq z) \geq t \} \\
	&=\inf\{z \in \R: \P(f\leq z \text{ and }k \leq z) \geq t \} \\
	&=\max\{ \inf\{z \in \R: \P(f\leq z) \geq t \}, k \} \\
	&=\max\{ F^{-1}_f(t), k \}.
	\end{split}
	\]
	\[
	\begin{split}
	F^{-1}_{\min\{f,k\}}(p)&=\inf\{z \in \R: \P(\min\{f,k\} \leq z) \geq p \} \\
	&=\inf\{z \in \R: \P(f\leq z \text{ or }k \leq z) \geq p \} \\
	&=\min\{ \inf\{z \in \R: \P(f\leq z) \geq p \}, k \} \\
	&=\min\{ F^{-1}_f(p), k \}.
	\end{split}
	\]
\end{proof}
\begin{proof}[Proof of (iii)]
	We use (ii) to derive the following identity
	\[
	\begin{split}
	F^{-1}_f(p) &= (F^{-1}_f)^+(p) + (F^{-1}_f)^-(p) \\
	&= \max\{F^{-1}_f(p),0\} + \min\{F^{-1}_f(p),0\} \\
	&= F^{-1}_{\max\{f,0\}}(p) + F^{-1}_{\min\{f,0\}}(p) \\
	&= F^{-1}_{f^+}(t) + F^{-1}_{f^-}.
	\end{split}
	\]
	The result now follows from the definition of $f^X$.
\end{proof}
\begin{proof}[Proof of (iv)]
	We wish to prove that the set 	
	\[A=\{ W \in \Omega: F^{-1}_X F^X X(W) \neq X(W) \}\]
	is null.

	We recall that
	\begin{equation}
	F^{-1}_X F_X( x) \leq x \text{ and } F_X F^{-1}_X(p) \geq p
	\label{eq:ffInvInequalities}
	\end{equation}
	for all $x\in \R$ and $p\in[0,1]$. We note that since $F_X$ is increasing this first inequality
	implies that
	\[
	F_X (F^{-1}_X F_X( x )) \leq F_X(x) 
	\]
	and the second implies
	\[
	F_X F^{-1}_X (F_X( x )) \geq F_X(x).
	\]
 	We deduce
	\begin{equation}
	F_X F^{-1}_X F_X( x ) = F_X(x).
	\label{eq:tripleFEquality}
	\end{equation}
	
	Suppose that $F^{-1}_X$ is continuous at $F_X X(W) \in [0,1]$ then
	\[
	\begin{split}
	F^{-1}_X F_X X(W) 
	&= \inf \{ F^{-1}_X(q) \mid q>F_X X(W) \} \\
	&= \inf \{ \inf \{ x \mid F_X(x)\geq q \} \mid q>F_X X(W)  \} \\
	&= \inf \{ x \mid F_X(x)>F_X X(W) \} \\
	&\geq X(W). \\
	\end{split}
	\]
	But by \eqref{eq:ffInvInequalities}, $F^{-1}_X F_X X (W) \leq X(W)$
	for all $W\in \Omega$. So if $F^{-1}_X$ is continuous at $F_X X(W)$ then
	$F^{-1}_X F_X X(W)=X(W)$, so $W \notin A$. Let $P$ denote the set of discontinuities of $F^{-1}_X$. We have shown:
	\[
	A \subseteq \bigcup_{p \in P} \{ \omega \mid F^{-1}_X F_X X(\omega) \neq X(\omega) \text{ and } F_X X(\omega)=p \}.
	\]
	Since $F^{-1}_X$ is monotone, $P$ is countable. Thus we can find a countable
	set $\{\omega_1, \omega_2, \ldots\}$ of elements of $\Omega$ such that
	\begin{equation}
	\begin{split}
	A &\subseteq \bigcup_{\omega_i} \{ \omega \mid F^{-1}_X F_X X(\omega) \neq X(\omega) \text{ and } F_X X(\omega)=F_X X(\omega_i) \} \\
	&= \bigcup_{\omega_i} \{ \omega \mid F^{-1}_X F_X X(\omega_i) \neq X(\omega) \text{ and } F_X X(\omega)=F_X X(\omega_i) \} \\
	&= \bigcup_{\omega_i} A_i \\
	\end{split}
	\label{eq:aAsUnion}
	\end{equation}
	where
	\begin{equation}
	\begin{split}
	A_i:&= \{ \omega \mid F^{-1}_X F_X X(\omega_i) \neq X(\omega) \text{ and } F_X X(\omega)=F_X X(\omega_i) \} \\
	&= \{ \omega \mid F_X X(\omega) = F_X X(\omega_i) \}
	\setminus \{ \omega \mid F^{-1}_X F_X X(\omega_i) = X(\omega) \}.
	\end{split}
	\label{eq:aAsComplement}
	\end{equation}
	We now note that
	\begin{multline}
	\{ \omega \mid F_X X(\omega) = F_X X(\omega_i) \}
	=  \\
	\{ \omega \mid F_X X(\omega) \leq F_X X(\omega_i) \}
	\setminus
	\{ \omega \mid F_X X(\omega) < F_X X(\omega_i) \}
	\label{eq:firstSetAsComplement}
	\end{multline}
	and
	\begin{multline}
	\{ \omega \mid F^{-1}_X F_X X(\omega_i) = X(\omega) \}
	= \\
	\{ \omega \mid X(\omega) \leq F^{-1}_X F_X X(\omega_i)  \}
	\setminus
	\{ \omega \mid X(\omega) < F^{-1}_X F_X X(\omega_i)  \}.
	\label{eq:secondSetAsComplement}
	\end{multline}
	Now
	\begin{equation}
	\begin{split}
	X(\omega)<F^{-1}_X F_X X(\omega_i)
	&\implies
	X(\omega) < \inf\{ x : F_X(x) \geq F_X X(\omega_i) \} \\
	&\implies F_X X(\omega)< F_X X(\omega_i).
	\end{split}
	\label{eq:ffimplication1}
	\end{equation}
	Conversely we can use \eqref{eq:tripleFEquality} to see that
	\begin{equation}
	\begin{split}
	F_X X (\omega)<F_x X(\omega_i)
	&\implies F_X X(\omega)<F_X F_X^{-1} F_X X(\omega_i) \\
	&\implies X(\omega)<F_X^{-1} F_X X(\omega_i)
	\end{split}
	\label{eq:ffimplication2}
	\end{equation}
	since $F_X$ is increasing. Together \eqref{eq:ffimplication1}
	and \eqref{eq:ffimplication2} imply
	\begin{equation}
	\{ \omega \mid F_X X(\omega) < F_X X(\omega_i) \} =
	\{ \omega \mid X(\omega) < F^{-1}_X F_X X(\omega_i)  \}.
	\label{eq:complementEquality}
	\end{equation}
	We use \eqref{eq:firstSetAsComplement}, \eqref{eq:secondSetAsComplement}
	and \eqref{eq:complementEquality} to rewrite \eqref{eq:aAsComplement} as
	\begin{equation}
	\begin{split}
	A_i = \{ \omega \mid F_X X(\omega) \leq F_X X(\omega_i) \}
	\setminus \{ \omega \mid X(\omega) \leq F^{-1}_X F_X X(\omega_i)  \}.
	\end{split}
	\label{eq:aAsPMinusP}
	\end{equation}
	Let $L_i=\{\omega \mid X(\omega) \leq F^{-1}_X F_X X(\omega_i) \}$.
	We use \eqref{eq:tripleFEquality} to compute that
	\begin{equation}
	\P( \omega \in L_i ) = F_X F^{-1}_X F_X X(\omega_i) = F_X X(\omega_i).
	\label{eq:sizeofL}
	\end{equation}
	
	Let $R_i=\{\omega \in \Omega \mid F_X X(\omega) = F_X X(\omega_i) \}$. We know $R_i$ is non empty since it contains $\omega_i$. Therefore we may
	choose a sequence $v^j_i$ in $R_i$ such that $X( v^j_i)$ is increasing and has limit equal to $\sup_{v \in R_i} X(v)$. Moreover if this supremum is obtained
	we may assume that the sequence $X(v^j_i)$ obtains its limit.
	
	\begin{equation}
	\begin{split}
	\{ \omega \mid F_X X(\omega)\leq F_X X(\omega_i) \}
	&= \bigcup_x \{ \omega \mid X(\omega) \leq x \text{ and } F_X(x) \leq F_X X(\omega_j) \} \\
	&= \bigcup_j \{ \omega \mid X(\omega) \leq X(v^j_i) \text{ and } F_X X(v^j_i) \leq F_X X(\omega_j) \} \\
	&= \bigcup_j \{ \omega \mid X(\omega) \leq X(v^j_i) \} \\
	&=\bigcup_j V^j_i
	\end{split}
	\label{eq:vjDefinition}
	\end{equation}
	Where $V^j_i:= \{ \omega \mid X(\omega) \leq X(v^j_i) \}$.
	We now compute that
	\begin{equation}
	\P(\omega \in V^j_i)=F_X X(v^j_i)=F_X X(\omega_i),
	\label{eq:sizeOfVj}
	\end{equation}
	since $v^j_i \in R_i$.
	
	Since $F^{-1}_X F_X X(\omega_i)=\inf \{X(\omega) \mid \omega \in R_i\} \leq X(v^j_i) $ we see that $L_i \subseteq V^j_i$. Hence $\P(\omega \in V^j_i \setminus L_i ) =
	\P(\omega \in V^j_i)-\P(\omega\in L_i )=0$, using \eqref{eq:sizeOfVj} and \eqref{eq:sizeofL}. By \eqref{eq:aAsUnion}, \eqref{eq:aAsPMinusP} and  \eqref{eq:vjDefinition} we have
	\[
	A \subseteq \bigcup_i \bigcup_j (V^j_i \setminus L_i).
	\]
	So $A_i$ is a countable union of null sets and hence is null.
\end{proof}
\begin{proof}[Proof of (v)]
	We define a generalized inverse for $G$ by
	\[
	G^{-1}(y) = \sup\{ x \in \R \mid G(x) \leq y \}.
	\]
	We define a function $\tilde{G}$ by
	\[
	\tilde{G}(x) = \inf\{ G(x^\prime) \mid x^\prime \geq x \}.
	\]
	We see that $\tilde{G}(x)=G(x)$ except possibly at the discontinuities of $G$.
	
	We note that
	\begin{equation}
	\begin{split}
	\tilde{G}(x) \leq y
	&\iff \inf \{ G(x^\prime) \mid x^\prime \geq x \} \leq y \\
	&\iff \exists x^\prime \text{ with } G(x^\prime)\leq y \text{ and }x^\prime \geq x \\
	&\iff x \leq \sup \{ x^\prime \mid G(x^\prime) \leq y \} \\
	&\iff x \leq G^{-1}(y).
	\label{eq:generalizedInverseIdentity}
	\end{split}
	\end{equation}
	
	We define $\tilde{g}(\omega)=\tilde{G}X^X(\omega)$. $G$ is monotone so only has
	a countable number of discontinuities. Let $D$ denote the set of discontinuities of $G$. Then $\tilde{G}(x)=G(x)$ unless $x \in D$. So
	the set of $\omega$ for which $\tilde{g}(\omega)\neq g(\omega)$ is
	contained in $X^{-1}(D) \cup A$ where $A$ is the null set defined in (iv). By the continuity of the distribution of $X$,
	$X^{-1}(x)$ is a null set for all $x$. Hence $\tilde{g}=g$ almost surely.
	
	We now wish to calculate $g_X(W)$ for $W \in \Omega$. In the calculation below, $W$ should be thought of as fixed and $\omega$
	should be thought of as a random scenario. So, for example $\P(X(\omega)\leq X(W))=F_X(X(W))$.
	\[
	\begin{split}
	g^X(W)
	&= F_{g}^{-1} (\P(X(\omega) \leq X(W ))) \\
	&= F_{\tilde{g}}^{-1} (\P(X(\omega) \leq X(W ))) \\
	&= \inf \{ x \mid F_{\tilde{g}}(x) \geq \P( X(\omega) \leq X(W )) \} \\
	&= \inf \{ x \mid F_{\tilde{G}X}(x) \geq \P( X(\omega) \leq X(W )) \} \\
	&= \inf \{ x \mid \P(\tilde{G}X(\omega)\leq x) \geq \P( X(\omega) \leq X(W)) \} \\
	&= \inf \{ x \mid \P(X(\omega)\leq G^{-1}x) \geq \P( X(\omega) \leq X(W)) \} 
	\qquad \text{ (by \eqref{eq:generalizedInverseIdentity})} \\
	&= \inf \{ x \mid F_X(G^{-1}x) \geq F_X(X(W)) \} \\
	&= \inf \{ x \mid G^{-1}x \geq X^X(W) \} \\
	&= \inf \{ x \mid X^X(W) \leq G^{-1}x \}
	\qquad \text{ (by \eqref{eq:generalizedInverseIdentity})}
	\\
	&= \inf \{ x \mid \tilde{G}(X^X(W)) \leq x \} \\
	&= \tilde{G} X^X(W) \\
	&= \tilde{g}(W).
	\end{split}
	\]
	Hence $g^X=g$ almost surely.
\end{proof}

\begin{lemma}
	\label{lemma:hardyLittlewood}
	If $f,g \in L^0(\Omega; \R)$ and:
	\begin{enumerate}[(i)]
		\item $f(\omega)\geq k$ for some $k\in \R$;
		\item $g \geq 0$;
		\item $\int_\Omega g \, \ed \P < \infty$ ;
		\item $X$ has a continuous distribution;
	\end{enumerate}
	then
	\[
	\int_\Omega f g \, \ed \P \leq \int_\Omega f^X g^X \, \ed \P \leq \infty.
	\]
\end{lemma}
\begin{proof}
	(Note: this proof is modelled on the proof of the Hardy--Littlewood inequality for ``symmetric decreasing rearrangements''.)
	
	Since $f\geq k$ we have the ``layer-cake'' representation of $f$
	\[
	f(\omega) = k + \int_0^\infty {\mathbbm{1}}_{L(f,x+k)}(\omega) \, \ed x
	\]
	where
	\[
	L(f,t):=\{ \omega \mid f(\omega)>t \}.
	\]
	We also have
	\[
	g(\omega) = \int_0^\infty {\mathbbm{1}}_{L(g,x)}(\omega) \, \ed x.
	\]
	
	We note that for any random variable $h$
	\[
	\begin{split}
	L(h^X,x) &= \{ \omega \mid h^X(\omega) > x \} \\
	&= \{ \omega \mid F^{-1}_h( \P (X \leq X(\omega )) > x \} \\
	&= \{ \omega \mid \P (X \leq X(\omega )) > F_h(x) \}.
	\end{split}
	\]
	Hence for any $h_1$, $h_2$, $x_1$, $x_2$
	either
	\begin{equation}
	L(h_1^X,x_1) \subseteq L(h_2^X,x_2) \quad \text{or} \quad 
	L(h_2^X,x_2) \subseteq L(h_1^X,x_1).
	\label{eq:subsetAlternative}
	\end{equation}
	We also note that
	\[
	\P( L(h,x) ) = \P( h(\omega) > x ) = 1-F_h(x).
	\]
	In particular $\P(L(h,x))$ only depends upon the distribution of $h$
	and hence $\P(L(h^X, x))=\P(L(h,x))$ by Lemma \ref{lemma:basicPropertiesOfRearrangement}.
	
	We now compute: 
	\begin{equation}
	\begin{split}
	\E^\P({\mathbbm 1}_L(f^X,x+k)(\omega) {\mathbbm 1}_L(g^X,y)(\omega))
	&= \P(L(f^X,(x+k)) \cap  L(g^X,y)) \\
	&= \min\{ \P(L(f^X,(x+k))),  \P(L(g^X,y)) \}  \quad \text{by \eqref{eq:subsetAlternative}} \\
	&= \min\{ \P(L(f,(x+k))),  \P(L(g,y)) \} \\
	&\geq \P(L(f,(x+k)) \cap L(g,y)) \\
	&= \E^\P({\mathbbm 1}_L(f,x+k)(\omega) {\mathbbm 1}_L(g,y)(\omega)). \\
	\label{eq:layerCakeInequality}
	\end{split}
	\end{equation}
	
	Using the fact that $f$ is bounded below and $\int_\Omega g \, \ed \P < \infty$ we
	deduce that
	\[
	\int_{\Omega} (f g)^- \, \ed \P > -\infty.
	\]
	By Lemma \ref{lemma:basicPropertiesOfRearrangement}, $f^X$ is also bounded
	below and $g=g^X$ in distribution so $\int_{\Omega} g^X \ed \P < \infty$. Hence
	\[
	\int_{\Omega} (f^X g^X)^- \, \ed \P > -\infty.
	\]
	Therefore we may use the layer-cake
	representations of $f$, $g$, $f^X$ and $g^X$ together with Fubini's theorem
	and \eqref{eq:layerCakeInequality} to compute:
	\[
	\begin{split}
	\int_\Omega f g \, \ed \P &= k \int_{\Omega} g \, \ed \P + \int_{\Omega}\int_\R
	\int_\R {\mathbbm 1}_L(f,x+k)(\omega){\mathbbm 1}_L(g,y)(\omega)
	\ed x \, \ed y \, \ed \P \\
	&= k \int_{\Omega} g \, \ed \P + \int_\R
	\int_\R \int_{\Omega} {\mathbbm 1}_L(f,x+k)(\omega){\mathbbm 1}_L(g,y)(\omega) \, \ed \P \, \ed x \, \ed y \\
	&\leq k \int_{\Omega} g \, \ed \P + \int_\R
	\int_\R \int_{\Omega} {\mathbbm 1}_L(f^X,x+k)(\omega){\mathbbm 1}_L(g^X,y)(\omega) \, \ed \P \, \ed x \, \ed y \\
	&= \int_{\Omega} f^X g^X  \, \ed \P.
	\end{split}
	\]
\end{proof}
\begin{lemma}
	\label{lemma:hardyLittlewoodEnhanced}
	If $f,g \in L^0(\Omega; \R)$ and:
	\begin{enumerate}[(i)]
		\item $\int f g \, \ed \P > -\infty$;
		\item $g \geq 0$;
		\item $\int_\Omega g \, \ed \P$ exists;
		\item $X$ has a continuous distribution;
	\end{enumerate}
	then
	\[
	-\infty < \int_\Omega f g \, \ed \P \leq \int_\Omega f^X g^X \, \ed \P \leq \infty.
	\]
\end{lemma}
\begin{proof}
	In this proof, given a real $k$ and random variable $f$, we will write $f_k$ as an abbreviation for the random variable $\max\{f(\omega),k\}$. Lemma \ref{lemma:basicPropertiesOfRearrangement} tells us $(f_k)^X=(f^X)_k$ so we may write $f_k^X$ without ambiguity.
	
	We know $\int (f g)^- \ed \P> -\infty$. Since $g \geq 0$, $(f g)^-=f^- g$, hence $\int f^- g \, \ed \P > -\infty$. Also since $g \geq 0$ we have for any $k \in \R$
	\[
	-\infty < \int_\Omega f^- g \, \ed \P \leq \int_\Omega f^-_k g  \, \ed \P
	\]
	By Lemma \ref{lemma:hardyLittlewood} we then have
	\[
	-\infty < \int_\Omega f^- g \, \ed \P \leq \int_\Omega f^-_k g  \, \ed \P \leq \int_\Omega (f^-)_k^X g^X \, \ed \P \ \ \mbox{for all} \ k.
	\]
	
	As $k\to -\infty$, $f^-_k(\omega) \downarrow f^-(\omega)$ and 
	$(f^-)^X_k(\omega) \downarrow (f^-)^X(\omega)$ for all $\omega$.
	So by the Montone Convergence Theorem
	\[
	-\infty < \int_\Omega f^- g \, \ed \P \leq \int_\Omega (f^-)^X g^X \, \ed \P.
	\]
	Lemma \ref{lemma:hardyLittlewood} also tells us that
	\[
	0 \leq \int_\Omega f^+ g \, \ed \P \leq \int_\Omega (f^+)^X g^X \, \ed \P. 
	\]
	Hence
	\[
	\begin{split}
	-\infty < \int_\Omega f g \, \ed \P
	&= \int_\Omega (f^+ + f^- )g \, \ed \P \\
	&\leq \int_\Omega ((f^+)^X + (f^-)^X) g^X \, \ed \P \\
	&\leq \int_\Omega (f^+ + f^-)^X g^X \, \ed \P \\
	&= \int_\Omega f^X g^X \ed P \leq \infty.
	\end{split}
	\]
\end{proof}

\begin{lemma}
	\label{lemma:sierpinski}
	If $(\Omega, {\cal F}, \P)$ is non-atomic and $Q \in L^0(\Omega; \R)$ is a random variable on $\Omega$, then there exists a
	uniform random variable
	$X \in L^0(\Omega; \R)$ such that
	$Q(\omega)=F^{-1}_Q(X(\omega))$.
\end{lemma}
\begin{proof}
	Sierpi\'nski's theorem on non-atomic measures tells us that for all $0 \leq \alpha \leq 1$ there is a measurable set $E \subseteq \Omega$ of measure $\alpha$ \cite{sierpinski}.
	
	One deduces that there is a uniformly distributed random variable $U$ on $\Omega$. To see this, first partition $\Omega$ into two subsets of measure $\frac{1}{2}$, we will call this partition Level $1$. We define a random variable $X_1$ which is equal to $\frac{1}{2}$ on the first subset of Level $1$ and equal to $1$ on the second subset of Level $1$. Inductively we partition each subset of Level $n$ into two equally sized subsets and define $X_n$ by $X_{n-1}-\frac{1}{2^n}$ on the first subsets and $X_{n-1}$ on the second subsets. The distribution function of $X_n$ will be a step function from $0$ to $1$ with $2^n$ uniform steps. By construction, $X^n(\omega)$ is decreasing for each $\omega \in \Omega$.
	Hence $X^n$ converges pointwise, hence almost surely, hence in distribution to some random variable $X$. Thus $X$ must be uniformly distributed.
	
	At each point $x \in \R$ where there is a discontinuity of $F_Q$ consider the set $\Omega_x=F_Q^{-1}(x)$. This set has non-zero measure $F_Q(x)-F_Q^-(x)$ where $F_Q^-(x)$ is the left limit of the distribution function at $x$. Note that $F_Q(x)-F_Q^-(x)=\P(Q=x)$. Hence by the above we can find a measurable function $U_x$ on $\Omega_x$ taking values uniformly between $0$ and $1$. Define
	a random variable $X$  by
	\[
	X(\omega)=\begin{cases}
	F_Q(Q(\omega)) & \text{if $F_Q$ is continuous at $Q(\omega)$} \\
	F_Q(x) - U_x(\omega) \P(Q=x) & \text{if $F_Q$ is discontinuous at $x=Q(\omega)$}. \\	
	\end{cases}
	\]	
    Clearly $Q(\omega)=F^{-1}_Q(X(\omega))$.  We must show that
    $X$ is uniformly distributed, i.e.\  that $P(X \leq p)=p$ for all $p\in[0,1]$.
    
    Given $p \in [0,1]$ define $p^- = \sup( \hbox{Im} F_Q \cap (-\infty,p] )$
    and $p^+ = \inf( \hbox{Im} F_Q \cap [p,\infty) )$. We partition $\Omega$ into three sets $A$, $B$ and $C$ defined by
    \[
    \begin{split}
    A &= \{ \omega \mid F_Q(Q(\omega)) \leq p^- \} \\
    B &= \{ \omega \mid p^- < F_Q(Q(\omega)) \leq p^+ \} \\    
    C &= \{ \omega \mid p^+ < F_Q(Q(\omega)) \}. \\        
    \end{split}
    \]
    For all $\omega$, $X(\omega)\leq F_Q(Q(\omega))$. So if $\omega \in A$, 
    $X(\omega) \leq F_Q(Q(\omega)) \leq p^- \leq p$. Hence
    \begin{equation}
    \P( X \leq p \mid A ) = 1.
	\label{eq:conditionalProbA}
    \end{equation}
    If $\omega \in B$ then $X(\omega)=p^+ - U_x(\omega)(p^+ - p^-)$. So 
    \begin{equation}
    \P(  X \leq p \mid B ) = \frac{p - p^-}{p^+-p^-}.
	\label{eq:conditionalProbB}
    \end{equation}
    
    For all $\omega$, $X(\omega)\geq F_Q^-(Q(\omega))$.
    Since $F_Q^-$ is a left limit, $F_Q^-(Q(\omega))\geq F_Q(x)$ if $x\leq Q(\omega)$.  If $\omega \in C$ then $F_Q^{-1}(p^+) < Q(\omega)$. We deduce that if $\omega \in C$ then
    \begin{equation*}
    X(\omega) \geq F_Q^-(Q(\omega)) \geq F_Q(F_Q^{-1} (p^+)) = p^+ \geq p.
    \end{equation*}
    We deduce that
    \begin{equation*}
    \P( X \leq p \mid C ) = \P( X = p \mid C ).
    \end{equation*}
    and moreover this probablity is equal to $0$ unless $p^+=p$. We may assume that each $U_x$ takes values in $(0,1)$ so that $X(\omega)$ never equals $p^+$ unless we have that $p^+=p^-$ and $F$ is continuous at $x=F^{-1}p^+$. But in this case we find that $\P(x=p^+)=\P(Q=x)=\P(Q\leq x) - \P(Q <x ) = F_Q(x) - F^-_Q(x)=0$. So $\P(x=p^+)=0$ and hence
    \begin{equation}
	\P(  X(\omega ) \leq p \mid C ) = 0.
	\label{eq:conditionalProbC}
	\end{equation}
	
	Since $p^{\pm} \in \hbox{Im} F_Q$ we compute that
	\[
\P(F_Q(Q(\omega)) \leq p^{\pm} ) = 
\P(Q(\omega) \leq F^{-1}_Q (p^{\pm}) ) =
F_Q( F^{-1}_Q (p^{\pm}) )) = p^{\pm}.
\]	
	It follows that
	\begin{equation}
	\P(A) = p^-, \  P(A \cup B ) = p^+, \text{ hence } P(B)=p^+ - p^-.
	\label{eq:sizeOfSets}
	\end{equation}
	Since $A$, $B$ and $C$ give a partition of $\Omega$ we may combine
	equations \eqref{eq:conditionalProbA}, \eqref{eq:conditionalProbB},
	\eqref{eq:conditionalProbC} and \eqref{eq:sizeOfSets} to obtain
	\[
	\P(X(\omega)\leq p ) = p.
	\]
	So $X$ is uniformly distributed as claimed.
\end{proof}

\begin{proof}[Proof of Theorem \ref{theorem:rearrangement}]
	Take $Q$ to be $\frac{\ed \Q}{\ed \P}$ in Lemma \ref{lemma:sierpinski} to find $X$ uniformly distributed with $\frac{\ed \Q}{\ed \P}=F^{-1}_\frac{\ed \Q}{\ed \P} V$ almost surely.
	
	Suppose $f$ satisfies the constraints of \eqref{optimizationProblem}. We see that
	$-(-f)^X$ is equal to $f$ in distribution, hence $-(-f)^X \in \cal A$ if $f \in \cal A$.
	Furthermore
	\[
	\begin{split}
	-e^{rT} C &\leq \int_\Omega (-f) \frac{\ed \Q}{\ed \P} \, \ed \P \quad \text{ by \eqref{optimizationProblem},} \\
	&\leq \int_\Omega (-f)^X \left(\frac{\ed \Q}{\ed \P}\right)^X \, \ed \P \quad \text{ by Lemma \ref{lemma:hardyLittlewoodEnhanced},} \\
	&= \int_\Omega (-f)^X \frac{\ed \Q}{\ed \P} \, \ed \P \leq \infty \quad \text{ by Lemma
		\ref{lemma:basicPropertiesOfRearrangement}.}
	\end{split}
	\]
	So $-(-f)^X$ satisfies the constraints of \eqref{optimizationProblem}.
	
	Finally we note that
	\[
	-(-f)^X (\omega) = -F^{-1}_{-f} F_X X (\omega) = -F^{-1}_{-f} X (\omega)
	= (1-F_{f})^{-1} X (\omega).
	\]
	
	The result now follows by taking $U=1-X$.
\end{proof}

\end{document}